\documentclass{article}
\usepackage{epsfig}
\usepackage{amsfonts}
\usepackage{amsmath}
\usepackage{amsthm}
\usepackage{graphicx}
\usepackage{amssymb}
\usepackage{float}
\usepackage{subfigure}
\usepackage{verbatim}
\usepackage{color}
\usepackage{soul}
\usepackage[margin=1.5in]{geometry}

\newtheorem{thm}{Theorem}
\newtheorem{definition}{Definition}
\newtheorem{assumption}{Assumption}

\begin{document}
\title{Dynamic phenotypes as criteria for model discrimination: fold-change
detection in \emph{R. sphaeroides}
chemotaxis}
\author{Abdullah Hamadeh, Brian Ingalls, Eduardo Sontag}
\date{}
\maketitle

\begin{abstract}
The chemotaxis pathway of the bacterium \emph{Rhodobacter
sphaeroides} has many similarities to the well-studied pathway in
\emph{Escherichia coli}.  It exhibits robust adaptation and has
several homologues of the latter's chemotaxis proteins.  Recent
theoretical results have been able to correctly predict that the
chemotactic response of \emph{Escherichia coli} exhibits the same
output behavior in response to scaled ligand inputs, a dynamic
property known as fold-change detection (FCD), or input-scale
invariance.  In this paper, we present theoretical assumptions on
the \emph{R. sphaeroides} chemotaxis sensing dynamics that can be
analytically shown to yield FCD behavior in a specific ligand
concentration range.  Based on these assumptions, we construct two
models of the full chemotaxis pathway that are able to reproduce
experimental time-series data from earlier studies. To test the
validity of our assumptions, we propose a series of experiments in
which our models predict robust FCD behavior where earlier models
do not.  In this way, we illustrate how a dynamic phenotype such
as FCD can be   used for the purposes of discriminating between
two models that reproduce the same experimental time-series data.
\end{abstract}

\section{Introduction}

Dynamic models of biological mechanisms are meaningful if they can
explain experimental data, make \emph{a priori} predictions of
biological behavior and be liable to invalidation through testing.

Although several competing models of a given mechanism can often
be made to reproduce experimental data through sufficient
parameterization and tuning, in many cases it is possible to
discriminate between such models by comparing the experimentally
observed output response and the simulated response to a
judiciously designed perturbation.   This paper is a study of the
use of a particular \emph{dynamic phenotype} for the purposes of
model discrimination. Dynamic phenotypes are distinctive,
qualitative, dynamic output responses that are robustly maintained
under a range of experimental conditions.

An example of a dynamic phenotype is adaptation, where a system
initially at steady-state reacts to input stimuli and then
restores its pre-stimulus equilibrium.  It has been shown that
integral control is the structural feature responsible for this
behavior \cite{Doyle}. Weber's law, whereby a system exhibits the
same maximal amplitude in its response to two different
inputs that are positive linear scalings of each other
\cite{SontagPNAS2010} is another example of a dynamic phenotype.

This paper deals with a third dynamic phenotype, termed fold
change detection (FCD) \cite{SontagPNAS2010}, or scale invariance.
A system is said to exhibit FCD if its output responses to two
different input stimuli that are positive linear scalings of
each other are identical (which makes this a stronger property
than Weber's law).

In a study \cite{SontagPNAS2010,SontagSIAM2011}, it was predicted
that the chemotaxis system of \emph{Escherichia coli}, modeled in
\cite{Shimizu2008}, would exhibit the FCD property, and these
predictions were later confirmed as accurate \cite{Shimizu2011}.
The key assumption of this model, which leads to FCD, is
the allosteric signaling structure of the methyl-accepting
chemotaxis protein receptors.


Although significantly more complex, the chemotaxis system of the bacterium \emph{Rhodobacter
sphaeroides} has many similarities to that of \emph{E.
coli}.  It features two, rather than one, sensory
clusters; one at the cell membrane and the other
in the cytoplasm.  Whilst the membrane cluster, as in \emph{E.
coli}, detects external ligand, it is as yet unknown
exactly what the cytoplasmic cluster senses
\cite{ArmitageNatureReview}.  Besides detecting
internalized ligand concentrations, it may also sense internal
signals, such as signals reporting the cell's metabolic state.  This bacterium also has multiple
homologues of the \emph{E. coli} chemotaxis proteins, which play roles similar to
those found in the latter, although the exact structure of their connectivity with the two sensory clusters and the flagellum is not known with certainty.  The CheA homologues transduce the receptor activity to the other
chemotaxis proteins through phosphotransfer, the CheR and CheB
homologues respectively methylate and demethylate receptors,
whilst the CheY proteins are believed to have a role in varying
the stopping frequency of the bacterium's single flagellum
\cite{Armitage_pnas_brake}.

Recent studies have used a model invalidation technique
to suggest possible connectivities for the CheY proteins
\cite{BMC} and the CheB proteins \cite{plos}.  However, upon simulation it
becomes evident that these models do not exhibit the FCD behavior observed in \emph{E. coli}.
This suggests the question: given the similarities between the two chemotaxis
pathways, does the \emph{R. sphaeroides} chemotaxis response show FCD as does that of \emph{E. coli}?

In this paper, we model the dynamics of the two \emph{R. sphaeroides} receptor clusters using the MWC allosteric
model \cite{MWC} that has been used to model the receptor activity in \emph{E. coli} in \cite{Berg,Keymer,Shimizu2008}.  We present a theorem that shows that if this is an accurate model of the receptor dynamics, then the receptor activities will exhibit FCD.  What is more, this observed behavior is robust to the connectivity between the chemotaxis proteins, the receptors and the flagellum.  To illustrate this point, we construct two models of the integrated \emph{R. sphaeroides} chemotaxis pathway based on our receptor dynamics assumptions, with each model featuring a different connectivity.  We show that, in addition to reproducing previously published experimental data, these models also display FCD in their flagellar responses in certain ligand concentration ranges.  Since flagellar outputs can be easily measured using tethered cell assays, we then suggest a series of experiments that can be used to test whether the models we present here are accurate compared to previously published models based on whether or not the flagellar response exhibits FCD.


This work therefore makes the case that qualitative dynamic behavior could be a
powerful property to test when discriminating between competing models.  A systematic way of model discrimination using this approach would start with the construction of a dynamic model that explains experimental data. The next step would be to use the model to mathematically identify experimentally implementable conditions under which the system can be expected to exhibit a certain dynamic phenotype.  The final step would be to experimentally implement those conditions and to compare the measured results against what is predicted \emph{in silico}.  In this way, two models which explain experimental data equally well can be discriminated using their dynamic phenotypes.

\subsection{Background} \label{Sec:background}

We can decompose the \emph{R. sphaeroides} chemotaxis pathway into three modules, as illustrated in Figure
\ref{Fig:rhodo}.  The sensing module includes two receptor clusters.  One of these resides at the cell membrane and senses the concentration of external ligands $L$, as illustrated in Figure \ref{Fig:rhodo_net}.   The other cluster resides within the cytoplasm and measures an internalized ligand concentration $\tilde{L}$.  Henceforth, the $\tilde{}$ notation will be used to denote signals associated with the cytoplasmic cluster.

\begin{figure}[h!]\begin{center}
\scalebox{0.6}{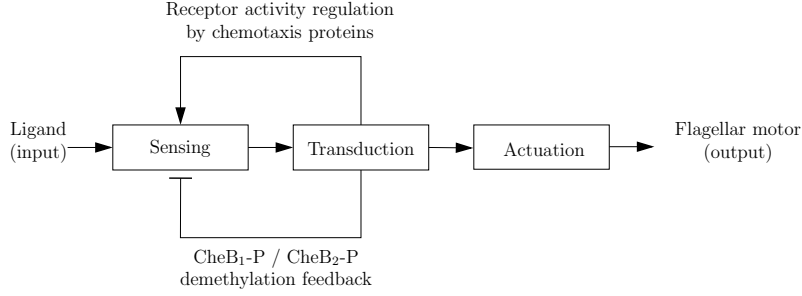} \caption{Schematic
of the \emph{R. sphaeroides} chemotaxis pathway.}
\label{Fig:rhodo}\end{center}\end{figure}

\begin{figure}[h!]\begin{center}
\scalebox{0.5}{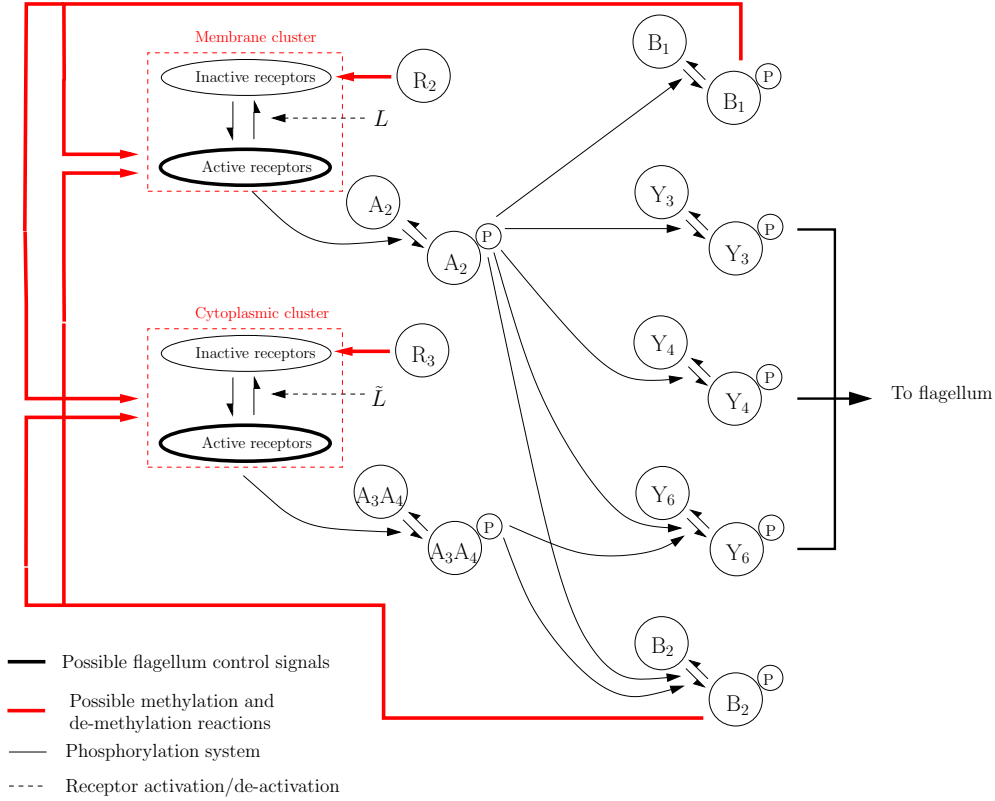} \caption{The \emph{R.
sphaeroides} signalling network.}
\label{Fig:rhodo_net}\end{center}\end{figure}

The dynamics of the two receptor clusters are modeled as two
first-order systems.  The membrane receptor cluster is assumed to
have state $m$ (its receptor methylation level) and output $a$
(the receptor activity level).  Similarly the cytoplasmic cluster
has methylation level $\tilde{m}$ as its state and its
activity level $\tilde{a}$ as its output.  The state-space
representation of this system is then
\begin{equation}\label{Eq:system_equations}\begin{array}{ll} \dot{m} & = F(a,w)\\
a & = G(m,L) \\
\dot{\tilde{m}} & = \tilde{F}(\tilde{a},\tilde{w})\\
\tilde{a} & = \tilde{G}(\tilde{m},\tilde{L}) \\
\end{array}\end{equation}
where $w,\tilde{w}$ are functions of the concentrations of the phorphorylated chemotaxis proteins within the cell.  These functions represent the interactions between the internal state of the cell and the receptors.  For example, $w$
and $\tilde{w}$ can represent the demethylation of the receptors
by the proteins CheB$_1$,CheB$_2$ or their methylation by the
proteins CheR$_2$,CheR$_3$.

The cytoplasmic cluster is believed
to integrate the extra-cellular ligand concentration $L$ with
internal cell signals.  We represent these internal cell signals by $u$, a function of the concentrations of the phorphorylated chemotaxis proteins. The signal $\tilde{L}$ in Figure
\ref{Fig:rhodo_net} is assumed to have the following relation with
the externally sensed ligand.
\begin{assumption}\label{Assump: input scaling}
The internalized ligand concentration $\tilde{L}$ is related to
the external ligand concentration $L$ through a linear, time
invariant filter
\[\begin{array}{l}\dot{\xi}=A\xi + B(u)L^{\nu}, \quad \xi \in \mathbb{R}^{n} \\  \tilde{L}=C\xi + D(u)L^{\nu}
\end{array}\] where $A \in \mathbb{R}^{n \times n}$, $B: \mathbb{R} \to \mathbb{R}^{n}$, $C\in \mathbb{R}^{1\times n}$, $D: \mathbb{R} \to \mathbb{R}$ and $\nu \in \mathbb{R}$.\end{assumption}

With Assumption \ref{Assump: input scaling}, the internalized
ligand concentration $\tilde{L}$ can represent a variety of
signals, including, for example, a static map that combines the
externally sensed ligands $L$ with the internal chemotaxis protein
signals $u$, or it can be a phase-delayed version of $L$ or even,
to allow for a degree of possible cooperativity, a power of $L$.

In the transduction sub-system of Figure \ref{Fig:rhodo}, auto-phosphorylation of the chemotaxis protein CheA$_2$
is accelerated by the membrane cluster activity, whilst that of
the CheA$_{3}$A$_{4}$ complex is catalyzed by cytoplasmic cluster
activity (as shown in Figure \ref{Fig:rhodo_net}).  The proteins
CheY$_3$, CheY$_4$, CheY$_6$, CheB$_1$ and CheB$_2$ all compete
for phosphoryl groups from CheA$_{2}$, whilst CheB$_{2}$ and
CheY$_6$ do so from CheA$_{3}$A$_{4}$.  The reaction rates for all
of these phosphorylations are given in \cite{Porter,BMC}.  We
represent this phosphotransfer network as a general nonlinear
system, with state vector \[\mathbf{x}=\left[\begin{array}{ccccccc}A_{2_{p}} &Y_{3_{p}}
&Y_{4_{p}} &(A_{3}A_{4})_{p} &Y_{6_{p}}  & B_{1_{p}} & B_{2_{p}}
\end{array}\right]^{T}\]
the individual states being the concentrations of the phosphorylated chemotaxis proteins.
The transduction system takes as its inputs
the receptor activities $a,\tilde{a}$:
\begin{equation}\label{Eq: signal transduction}\begin{array}{l}\dot{\mathbf{x}}=H(\mathbf{x},a,\tilde{a})\end{array}\end{equation}
where $H(\mathbf{x},a,\tilde{a})$ is given by the ODEs (9)-(15) in \cite{BMC}.

%
%

The outputs of this system are signals
$w(\mathbf{x}),\tilde{w}(\mathbf{x}),u(\mathbf{x},a)$, which feed
back into the sensing subsystem, as described above.  The interconnection of the phosphotransfer network (\ref{Eq: signal transduction}) with the receptor dynamics is illustrated in Figure \ref{Fig:upstream_downstream}, and the interconnection between the two subsystems can thus be written as \begin{equation}\label{Eq:combined odes}\begin{array}{l}\dot{m}=F(a,w(\mathbf{x})), \; a=G(m,L)\\
\dot{\tilde{m}}=\tilde{F}(\tilde{a},\tilde{w}(\mathbf{x})),  \; \tilde{a}=\tilde{G}(\tilde{m},\tilde{L})\\
\dot{\xi}=A\xi + B(u(\mathbf{x},a))L^{\nu}, \; \tilde{L}=C\xi + D(u)L^{\nu}  \\
\dot{\mathbf{x}}=H(\mathbf{x},a,\tilde{a}) \\
\end{array}
\end{equation}

\begin{figure}[h!]\begin{center}
\scalebox{0.7}{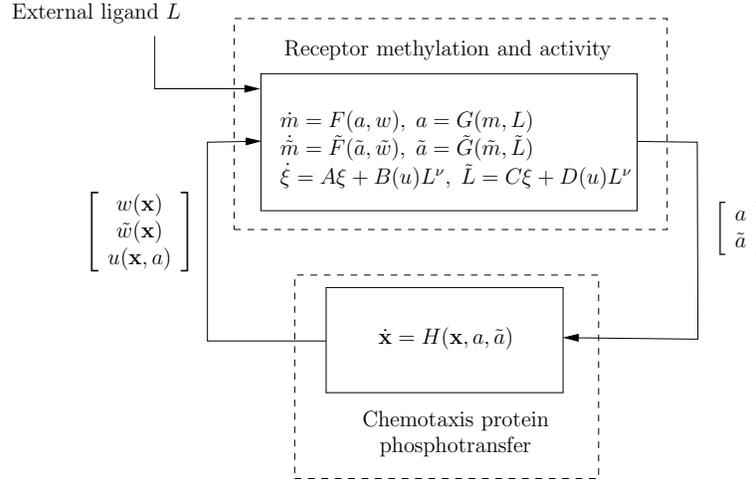} \caption{The
interconnection of the receptors' sensing dynamics with the
phosphotransfer network.}
\label{Fig:upstream_downstream}\end{center}\end{figure}

As shown in Figure \ref{Fig:rhodo_net}, the protein CheY$_6$-P,
possibly acting together with one or both of CheY$_3$-P and
CheY$_4$-P, is believed to bind with the flagellar motor proteins to
inhibit the flagellar rotation rate \cite{Armitage_pnas_brake}
(thus effectively coupling the signal transduction system to the
actuation system), though the precise mechanism through which this
is achieved is unknown.  An additional uncertainty lies in the
demethylation connectivity between the CheB proteins and the two
receptor clusters, though these questions have been the subjects
of several studies, \cite{BMC,plos,tindall}. Although how the CheB
and CheY proteins interact with the sensing and actuation modules
is not known for certain, it will be shown later in the paper that
under some mild assumptions, FCD can be exhibited by the bacterium
regardless of the exact structure of these connectivities.

\subsection{An MWC model of receptor dynamics}\label{Sec:Assumptions}
%
%
%
%

We employ an MWC-type allosteric model
for the receptor activities \cite{MWC}.  Such models have been proposed for
several bacterial chemotactic systems and have been found to be
consistent with experimental data \cite{Berg, Keymer,claus,Shimizu2008}.  The main assumptions of the model are that receptors are either active or inactive, and that ligands have a higher affinity for inactive receptors than for active receptors. Respectively, we denote by $a(t)$ and $\tilde{a}(t)$ the probabilities at time $t$ of a transmembrane and cytoplasmic receptor being active. For each receptor, this probability can be
approximated by the ratio of the Boltzmann factor of the active
state to the sum of the Boltzmann factors of all the states.
Therefore if, at time $t$, the free-energy state of the membrane
receptors is $E_{A}$ when active and $E_{I}$ when inactive, then
the activity of the membrane receptors is approximated by
\begin{equation}\label{Eq: act mem}a(t)= \frac{\exp(-E_{A})}{\exp(-E_{I})+\exp(-E_{A})}  = \frac{1}{1+\exp[-E_{\Delta}]}  \end{equation}
where $E_{\Delta}=E_{I} - E_{A}$ is the free energy difference
between the active and inactive states.

Similarly for the cytoplasmic receptors, the activity
$\tilde{a}(t)$ is dependent on their free-energy states when
active and inactive, respectively $\tilde{E}_{A}$ and
$\tilde{E}_{I}$:
\begin{equation}\label{Eq: act cyt}\tilde{a}(t)= \frac{\exp(-\tilde{E}_{A})}{\exp(-\tilde{E}_{I})+\exp(-\tilde{E}_{A})}  =   \frac{1}{1+\exp [ - \tilde{E}_{\Delta} ]}  \end{equation}
with $\tilde{E}_{\Delta}=\tilde{E}_{I} - \tilde{E}_{A}$. The
functions $E_{\Delta}$ and  $\tilde{E}_{\Delta}$ are assumed to
have the same structure and take the form
$E_{\Delta}=-[g_{m}(m)+g_{L}(L)]$ and
$\tilde{E}_{\Delta}=-[\tilde{g}_{m}(\tilde{m})+\tilde{g}_{L}(\tilde{L})]$
The functions $g_{m},\tilde{g}_{m}$ are dependent on the
methylation state of their respective receptors whilst the functions $g_{L},\tilde{g}_{L}$ quantify the effect of ligand
binding on the receptor-free energy difference of the receptors.  Following
\cite{Shimizu2008,claus,Shimizu2011} we make the assumption that each of
$g_{m},\tilde{g}_{m}$ is affinely dependent on the methylation
state of its respective receptor cluster:
\[g_{m}(m)=\alpha(m_{0}-m) \text{ and }
\tilde{g}_{m}(\tilde{m})=\tilde{\alpha}(\tilde{m}_{0}-\tilde{m})
\] where $\alpha=\tilde{\alpha}=2$ and $m_{0}=\tilde{m}_{0}=5$.

The binding of ligands to receptors leads to a loss of ligand translational entropy, proportional to the logarithm of the free ligand concentration \cite{Keymer,Shimizu2008}.  Due to the greater affinity of ligands to inactive receptors, this loss is greater in the case of ligands binding to active receptors.  We denote the dissociation constants between ligands and active transmembrane (cytoplasmic) receptors by $K_{A}$ ($\tilde{K}_{A}$), and between ligands and inactive transmembrane (cytoplasmic) receptors by $K_{I}$ ($\tilde{K}_{I}$), with $K_{A} \gg K_{I}$ and $\tilde{K}_{A} \gg \tilde{K}_{I}$ due to the different affinities.  From the \emph{E. coli}
chemotaxis literature, we adopt the values $K_{I}=\tilde{K}_{I}=18
\mu M$, $K_{A}=\tilde{K}_{A}=3 mM$.   As in \cite{Keymer}, the change in receptor free energies due to ligand binding to active transmembrane and cytoplasmic receptors is then, respectively, $-\ln(\frac{L}{K_{A}})$ and $-\ln(\frac{\tilde{L}}{\tilde{K}_{A}})$.  On the other hand the change in receptor free energy due to ligand binding to inactive transmembrane and cytoplasmic receptors is, respectively, $-\ln(\frac{L}{K_{I}})$ and $-\ln(\frac{\tilde{L}}{\tilde{K}_{I}})$.  The effect of this on the free energy differences $E_{\Delta}$, $\tilde{E}_{\Delta}$ between active and inactive receptors can be characterized, as in
\cite{Keymer,Shimizu2008}, as
\[\begin{split}&g_{L}(L)=\ln\left(1+\frac{L}{K_{I}}\right)-\ln\left(1+\frac{L}{K_{A}}\right)
\text{ and } \\ &
\tilde{g}_{L}(\tilde{L})=\ln\left(1+\frac{\tilde{L}}{\tilde{K}_{I}}\right)-\ln\left(1+\frac{\tilde{L}}{\tilde{K}_{A}}\right)\end{split}\]
for each cluster respectively.  Due to the differences in affinities, we note that $g_{L}(L)$ and $g_{L}(\tilde{L})$ are increasing functions of $L$ and $\tilde{L}$ respectively, which means that $a(t)$ and $\tilde{a}(t)$ are decreasing functions of $L$ and $\tilde{L}$ respectively.  The greater affinity of ligands for inactive receptors therefore has the effect of shifting the receptors towards the inactive state.

Note that in the ligand concentration range $K_{I} \ll L \ll
K_{A}$ and $\tilde{K}_{I} \ll \tilde{L} \ll \tilde{K}_{A}$, the
receptor activities can be approximated by
\begin{equation}\label{Eq:approx_act}\begin{split} & a=\frac{1}{1+\left[\exp(\alpha[m_{0}- m])\frac{L}{K_{I}}\right]} \text{ and } \\ & \tilde{a}=\frac{1}{1+\left[\exp(\tilde{\alpha} [\tilde{m}_{0}-\tilde{m}])\frac{\tilde{L}}{\tilde{K}_{I}}\right]} \end{split}\end{equation}

\section{Main results}

Following similar definitions in the literature,
\cite{SontagPNAS2010,SontagSIAM2011} we give the following
definition of fold-change detection for the \emph{R. sphaeroides}
chemotaxis pathway.

\begin{definition}\label{Def: FCD} The \emph{R. sphaeroides} chemotaxis system (\ref{Eq:system_equations}) exhibits fold change detection (FCD) in response to a sensed ligand input signal $L(t)$ if its receptor activities $a(t),\tilde{a}(t)$, initially at a steady state corresponding to $L(0)$, are independent of linear scalings $p>0$ of the input $L(t)$.\end{definition}

Note that the chemotaxis protein phosphorylation network (\ref{Eq: signal transduction}) takes as its
sole inputs the signals $a$ and $\tilde{a}$.  For this reason, Definition \ref{Def: FCD} implies that if the system
(\ref{Eq:combined odes}) exhibits FCD in its activities, it
also exhibits FCD in the concentration of its phosphorylated
chemotaxis proteins (the elements of the vector $\mathbf{x}$).  The
bacterium's flagellar behavior would also be expected to exhibit
FCD as the flagellum rotation rate is a function of the
phosphorylated CheY$_{3}$, CheY$_{4}$ and CheY$_{6}$
concentrations.  Before giving the main result, we make the
following assumption on the chemotaxis system dynamics.

\begin{assumption}\label{Assump: unique fixed point}
The system (\ref{Eq:combined odes}) has a unique steady state for
any given $L$.\end{assumption}

\begin{thm}\label{Thm: FCD} Under Assumptions \ref{Assump: input scaling} and \ref{Assump: unique fixed point}, and under approximation (\ref{Eq:approx_act})
the chemotaxis system (\ref{Eq:combined odes}), with steady state
initial conditions, will exhibit FCD in its activities
$a,\tilde{a}$ for ligand inputs in the range $K_{I} \ll L \ll
K_{A}$ and $\tilde{K}_{I} \ll \tilde{L} \ll \tilde{K}_{A}$, in the
sense of Definition \ref{Def: FCD}.
\end{thm}

\begin{proof}
In the following, we assume that all ligand
concentrations lie in the ranges $K_{I} \ll L \ll K_{A}$ and
$\tilde{K}_{I} \ll \tilde{L} \ll \tilde{K}_{A}$, and therefore
approximation (\ref{Eq:approx_act}) holds.  The proof is based on
the existence of equivariances \cite{SontagSIAM2011}.

Suppose that in response to an external ligand input signal
$L=L_{1}(t)$, the system (\ref{Eq:combined odes}), initially at a
steady state corresponding to $L=L_{1}(0)$, exhibits a solution
\[\left[\begin{array}{c}m \\ \tilde{m} \\ \textbf{x} \\ \xi
\end{array}\right]=\left[\begin{array}{c}m_{1}(t) \\ \tilde{m}_{1}(t) \\ \textbf{x}_{1}(t) \\ \xi_{1}(t)
\end{array}\right]=\mathbf{m}_{1}(t)\] and outputs $a_{1}(t)=G(m_{1}(t),L_{1}(t))$, $\tilde{L}_{1}(t)= C\xi_{1}+D(u(\mathbf{x}_{1},a_{1}))L_{1}^{\nu}$, $\tilde{a}_{1}=\tilde{G}(\tilde{m}_{1}(t),\tilde{L}_{1}(t))$.
Now if the ligand input is scaled to $L=L_{2}(t)=pL_{1}(t)$, where
$p > 0$, and if the initial state corresponds to $L=L_{2}(0)$, then \begin{equation}\label{Eq:other soln}\mathbf{m}_{2}(t)=\left[\begin{array}{c}m_{2}(t) \\
\tilde{m}_{2}(t)
\\ \textbf{x}_{2}(t) \\ \xi_{2}(t)
\end{array}\right]=\left[\begin{array}{c}m_{1}(t) + \frac{1}{\alpha} \log p \\ \tilde{m}_{1}(t) +
\frac{1}{\tilde{\alpha}} \log p^{\nu} \\ \textbf{x}_{1}(t) \\
p^{\nu}\xi_{1}(t) \end{array}\right]\end{equation} is a solution
of (\ref{Eq:combined odes}) since, under approximation
(\ref{Eq:approx_act}), the outputs are then
\[\begin{array}{l}

\begin{array}{ll}a_{2} & =G(m_{2},L_{2}) \\ & =G(m_{1}+\frac{1}{\alpha}\log
p,pL_{1})=G(m_{1},L_{1})=a_{1}\end{array} \\

\vspace{1mm}

\\

\begin{array}{ll} \tilde{L}_{2} & =C\xi_{2}+D(u(\mathbf{x}_{2},a_{2}))L_{2}^{\nu} \\
& =p^{\nu}C\xi_{1}+D(u(\mathbf{x}_{1},a_{1}))p^{\nu}L_{1}^{\nu}=p^{\nu}\tilde{L}_{1}\end{array} \\

\vspace{1mm}

\\

\begin{array}{ll} \tilde{a}_{2}& =\tilde{G}(\tilde{m}_{2},\tilde{L}_{2}) \\ & =\tilde{G}(\tilde{m}_{1}+\frac{1}{\tilde{\alpha}}\log
p^{\nu},p^{\nu}\tilde{L}_{1})=\tilde{G}(\tilde{m}_{1},\tilde{L}_{1})=\tilde{a}_{1}
\end{array}
\end{array}
\]
which means that
\[\begin{split}\hspace{-1mm}\frac{d}{dt}\hspace{-1mm}\left[\hspace{-1mm}\begin{array}{c}m_{2} \\ \tilde{m}_{2} \\ \textbf{x}_{2} \\ \xi_{2}
\end{array}\hspace{-1mm}\right]\hspace{-1mm}  & =\hspace{-1mm} \frac{d}{dt}\hspace{-1mm}\left[\hspace{-1mm}\begin{array}{c}m_{1}(t) + \frac{1}{\alpha} \log p \\ \tilde{m}_{1}(t) +
\frac{1}{\tilde{\alpha}} \log p^{\nu} \\ \textbf{x}_{1}(t) \\
p^{\nu}\xi_{1}(t)
\end{array}\hspace{-1mm}\right]\hspace{-1mm} \\ & \hspace{-1mm}=\left[\hspace{-1mm}\begin{array}{c}F(a_{1},w(\mathbf{x}_{1})) \\
\tilde{F}(\tilde{a}_{1},\tilde{w}(\mathbf{x}_{1}))
 \\ H(\textbf{x}_{1},a_{1},\tilde{a}_{1}) \\
p^{\nu}[A\xi_{1} + B(u(\mathbf{x}_{1},a_{1}))L_{1}^{\nu}]
\end{array}\hspace{-1mm}\right] \\ & \hspace{-1mm}=\left[\hspace{-1mm}\begin{array}{c}F(a_{2},w(\mathbf{x}_{2})) \\ \tilde{F}(\tilde{a}_{2},\tilde{w}(\mathbf{x}_{2}))
 \\ H(\textbf{x}_{2},a_{2},\tilde{a}_{2}) \\
A\xi_{2} + B(u(\mathbf{x}_{2},a_{2}))L_{2}^{\nu}
\end{array}\hspace{-1mm}\right]\end{split}\]
which verifies the claim that (\ref{Eq:other soln}) is a solution
of (\ref{Eq:combined odes}) when the ligand input is
$L(t)=pL_{1}(t)$. Since the scaled inputs $L=L_{1}(t)$ and
$L=L_{2}(t)=pL_{1(t)}$ yield the respective output pairs
$a_{1},\tilde{a}_{1}$ and $a_{2},\tilde{a}_{2}$ and since
$a_{1}=a_{2}$ and $\tilde{a}_{1}=\tilde{a}_{2}$, it follows that
system (\ref{Eq:combined odes}) under Assumption \ref{Assump:
input scaling} exhibits fold change detection if the initial
conditions of the system are $\mathbf{m}_{1}(0)$ when $L=L_{1}$
and $\mathbf{m}_{2}(0)$ when $L=L_{2}(t)=pL_{1}(t)$.  In the
language of \cite{SontagSIAM2011}, we have proved that the mapping
$\mathbf{m}_{1} \mapsto \mathbf{m}_{2}$ is an equivariance
associated to scalar symmetries on inputs.

Now if the system has a unique fixed point for any given $L$, and
if $\mathbf{m}_{1}(0)$ is the fixed point when $L=L_{1}(0)$, then
$\mathbf{m}_{2}(0)$ is the fixed point when $L=L_{2}(0)=pL_{1}(0)$
since, if $\dot{\mathbf{m}}_{1}=\mathbf{0}$ when $L=L_{1}(0)$ then
$\dot{\mathbf{m}}_{2}=\mathbf{0}$ when $L=L_{2}(0)=pL_{1}(0)$.
Therefore if system (\ref{Eq:combined odes}) has a unique fixed
point for any given $L$, starts from steady state conditions and
yields solution $\mathbf{m}_{1}(t)$, then scaling $L$ by $p>0$ and
initiating the system from steady state conditions will cause the
system to yield the solution $\mathbf{m}_{2}(t)$ and thereby
exhibit FCD.
\end{proof}

\section{Two \emph{R. sphaeroides} chemotaxis models} \label{Sec: models}

There are several integrated \emph{R. sphaeroides} chemotaxis pathway models in the literature \cite{plos,BMC,tindall}.  In this section, we present two new models, differing from the previous ones in that their receptor dynamics are of the form (\ref{Eq:combined odes}) and satisfy the MWC model given in Section \ref{Sec:Assumptions}.  Both of the new models we present were fitted to experimental data available in \cite{plos} and were able to reproduce the gene deletion data in \cite{BMC,plos}.

Each of the models presented satisfies the assumptions of Section \ref{Sec:Assumptions} and thereby exhibits FCD in the ligand range
$K_{I} \ll L \ll K_{A}$ and $\tilde{K}_{I} \ll \tilde{L} \ll \tilde{K}_{A}$. The demethylating feedback structure for the models is restricted to that in \cite{plos}, which proposed an asymmetric feedback structure wherein CheB$_1$ demethylates both clusters and CheB$_{2}$ demethylates the cytoplasmic cluster, although a model with any feedback connectivity is capable of exhibiting FCD under the assumptions we make.

The structural differences between the models lie in the signal $\tilde{L}$, which captures how external ligands are transduced to the cytoplasmic cluster. These models illustrate the point that, despite the differences in their internal connectivities, FCD behavior is conserved under the assumptions above.

Following \cite{BarkaiLeibler1997,plos}, we make the assumptions that
CheB proteins demethylate active receptors, whilst CheR proteins
methylate inactive receptors, and that CheR proteins operate at
saturation.  The CheR$_{2}$ and CheR$_{3}$ protein concentrations
are therefore assumed to be constant and normalized to $1 \mu M$
each. Denoting by $R_{2},R_{3}$ the concentrations of CheR$_{2}$ and CheR$_{3}$ and by $B_{1_{p}},B_{2_{p}}$ the concentrations of phosphorylated chemotaxis proteins CheB$_{1}$, CheB$_{2}$, mass action kinetics give the following general form for $F, \tilde{F}$ in (\ref{Eq:combined odes})
\begin{equation}\label{Eq: methylation}\begin{array}{ll}\dot{m}=F(a,\tilde{a},w(\mathbf{x}))=k_{R}(1-a)R_{2}-k_{B_{1}}B_{1_{p}}a-k_{B_{2}}B_{2_{p}}a \\
\dot{\tilde{m}}=\tilde{F}(a,\tilde{a},\tilde{w}(\mathbf{x}))=\tilde{k}_{R}(1-\tilde{a})R_{3}-\tilde{k}_{B_{2}}B_{2_{p}}\tilde{a}\end{array}
\end{equation} where $k_{R}, \tilde{k}_{R}>0$ are methylation and $k_{B_{1}}, k_{B_{2}},\tilde{k}_{B_{2}}>0$ demethylation
rate constants. The probabilities of activity $a,\tilde{a}$ given by (\ref{Eq: act mem}), (\ref{Eq: act cyt}).  The models were
obtained by fitting the constants $k_{R}, \tilde{k}_{R},
k_{B_{1}},k_{B_{2}},\tilde{k}_{B_{2}}$ in (\ref{Eq: methylation}).

The experimentally measured output which was used to fit the model
is the flagellar rotation frequency $f$.   As shown in Figure
\ref{Fig:rhodo_net}, the CheY proteins control the rotation of the
flagellum, and this is believed to happen through inhibitory
binding \cite{BMC}.  The measured rotation frequencies to which we fit our
models varied between 0 Hz and a maximum of approximately 8 Hz. As
discussed in \cite{plos}, this maximum was very rarely exceeded,
and is therefore assumed to be a physical limit on how fast the
flagellum can rotate. As such, the rotation frequency is modeled
as the Hill function
\[f=-\frac{1}{0.125 + \phi(Y_{3_{p}},Y_{4_{p}},Y_{6_{p}})^{4}}\]
where
\[\phi(Y_{3_{p}},Y_{4_{p}},Y_{6_{p}})=0.012Y_{6_{p}}\frac{Y_{3_{p}}+Y_{4_{p}}}{0.1+Y_{3_{p}}+Y_{4_{p}}}\]
(the negative sign denotes anti-clockwise rotation).  In this way,
$f$ varies between 0 - 8 Hz, and decreases with increased
concentrations of phosphorylated CheY proteins.

\subsubsection{Model I}\label{Subsec:model I}


\noindent \textbf{Model structure:}  In this model the cytoplasmic
receptors are assumed to sense internalized ligands, the
concentrations of which are dependent on the external ligand
concentration $L$.  At the same time, as in \cite{plos}, we assume there to be some
interaction between the chemotaxis proteins CheY$_3$, CheY$_4$ and
the cytoplasmic cluster, and the function $\tilde{g}_{L}(\tilde{L})$ takes as its input  $\tilde{L}=\frac{10L}{10+Y_{3_{p}}+Y_{4_{p}}}$.  A schematic of this model is shown in Figure \ref{Fig:model_I_schematic}.



A simulation of the model together with the tethered cell trace to
which the model was fitted is shown in Figure
\ref{Fig:model_I_assay}.  For comparison, Figure
\ref{Fig:model_I_assay} additionally shows a simulation (with the
same ligand input) of the model suggested in \cite{plos}, which
was fitted to the same tethered cell assay.  The root mean squared
error between the output of Model I and the tethered cell assay is
0.88, which compares favorably to the corresponding error for the
model in \cite{plos}, which is 1.27.

This model would be expected to exhibit FCD in the ligand ranges
$K_{I} \ll L \ll K_{A} $ and $\tilde{K}_{I} \ll \tilde{L} \ll
\tilde{K}_{A}$. The latter range is equivalent to
\[\tilde{K}_{I}\left(1+0.1[Y_{3_{p}}+Y_{4_{p}}]\right) \ll L \ll
\tilde{K}_{A}\left(1+0.1[Y_{3_{p}}+Y_{4_{p}}]\right)\] and since
the total amounts of intracellular CheY$_{3}$ and CheY$_{4}$
(phosphorylated and un-phosphorylated) are $3.2 \mu$M and $13.2
\mu$M respectively, then according to this model, simulations should show FCD in the range $ 2.64 \tilde{K}_{I} \ll  L \ll
\tilde{K}_{A}$.  Figure \ref{Fig:model_I_fcd} shows that this is
indeed the case, with similar output traces obtained for the step
changes in $L$ from $L=1000 \mu$M to 200 $\mu$M and from $L=500
\mu$ M to $100 \mu$M.

%

\noindent \textbf{Model parameters:}
 $k_{R}=\tilde{k}_{R}=0.0045 \quad
k_{B_{1}}=\tilde{k}_{B_{2}}=2.116 \quad k_{B_{2}}=2.822$.

\begin{figure}[h!]\begin{center}
\scalebox{0.5}{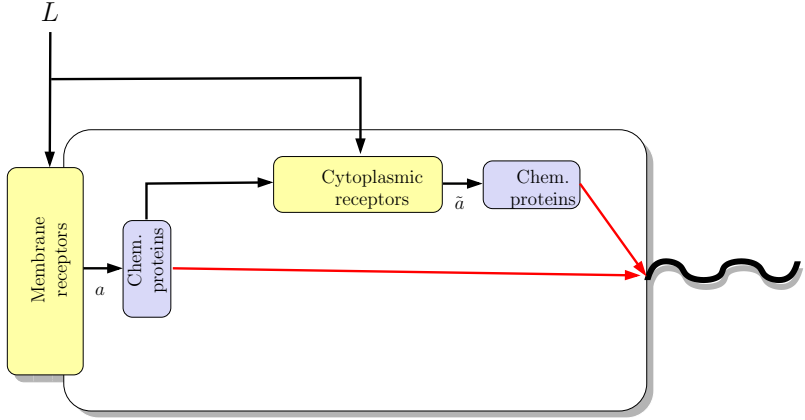}
\caption{Schematic of Model I.}
\label{Fig:model_I_schematic}\end{center}\end{figure}

\begin{figure}[h!]\begin{center}\includegraphics[width=0.7\textwidth]{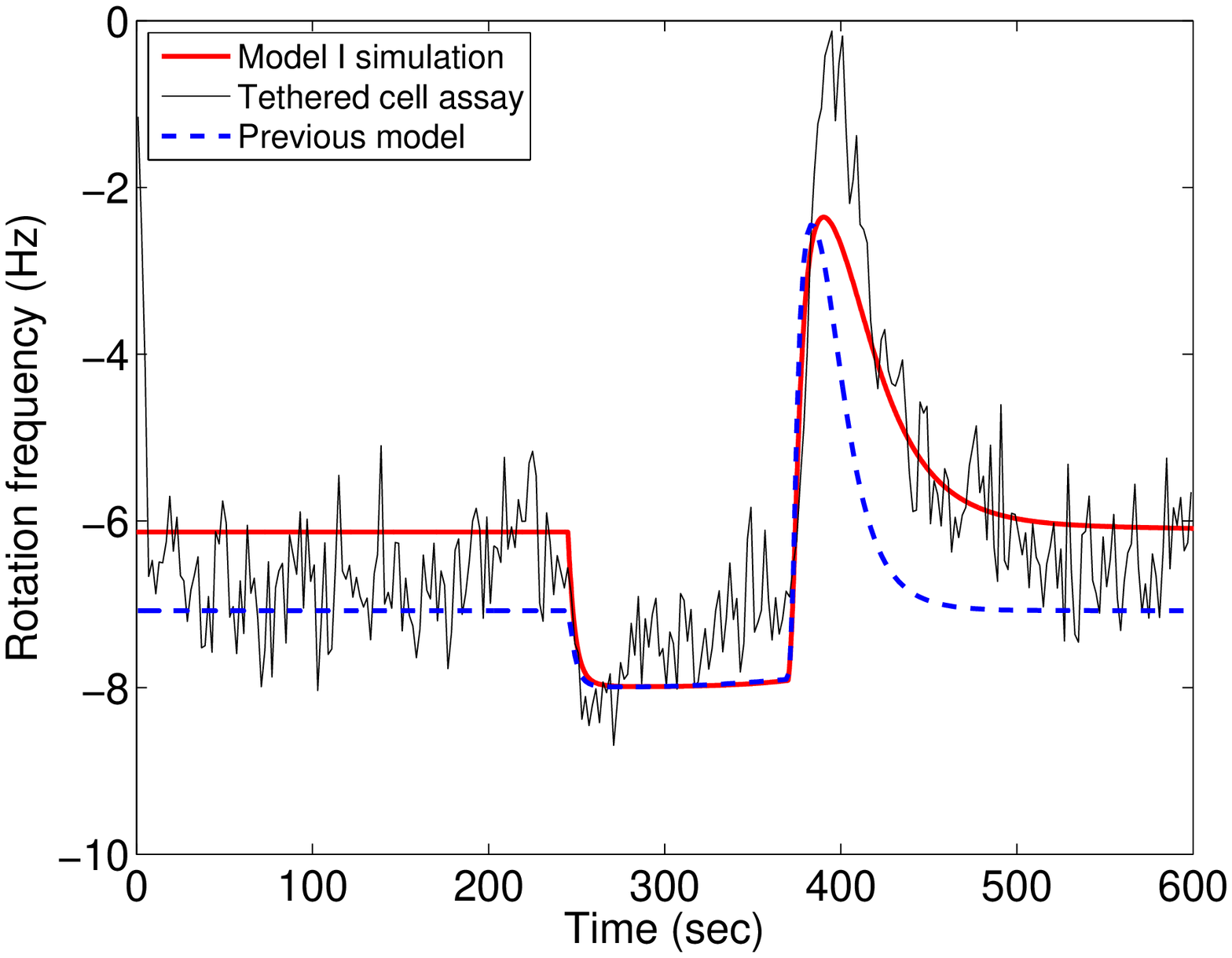}
\caption{Simulation of Model I (red) in response to a step rise
(at 245 seconds) and fall (at 370 seconds) in the ligand level $L$
from $L=0$ to $L=100$ and back to $L=0$, with a tethered cell
assay (black).  The dashed blue trace is a simulation of the
previously published model in \cite{plos} subject to the same
ligand input.}\label{Fig:model_I_assay}\end{center}\end{figure}

\begin{figure}[h!]\begin{center}\includegraphics[width=0.7\textwidth]{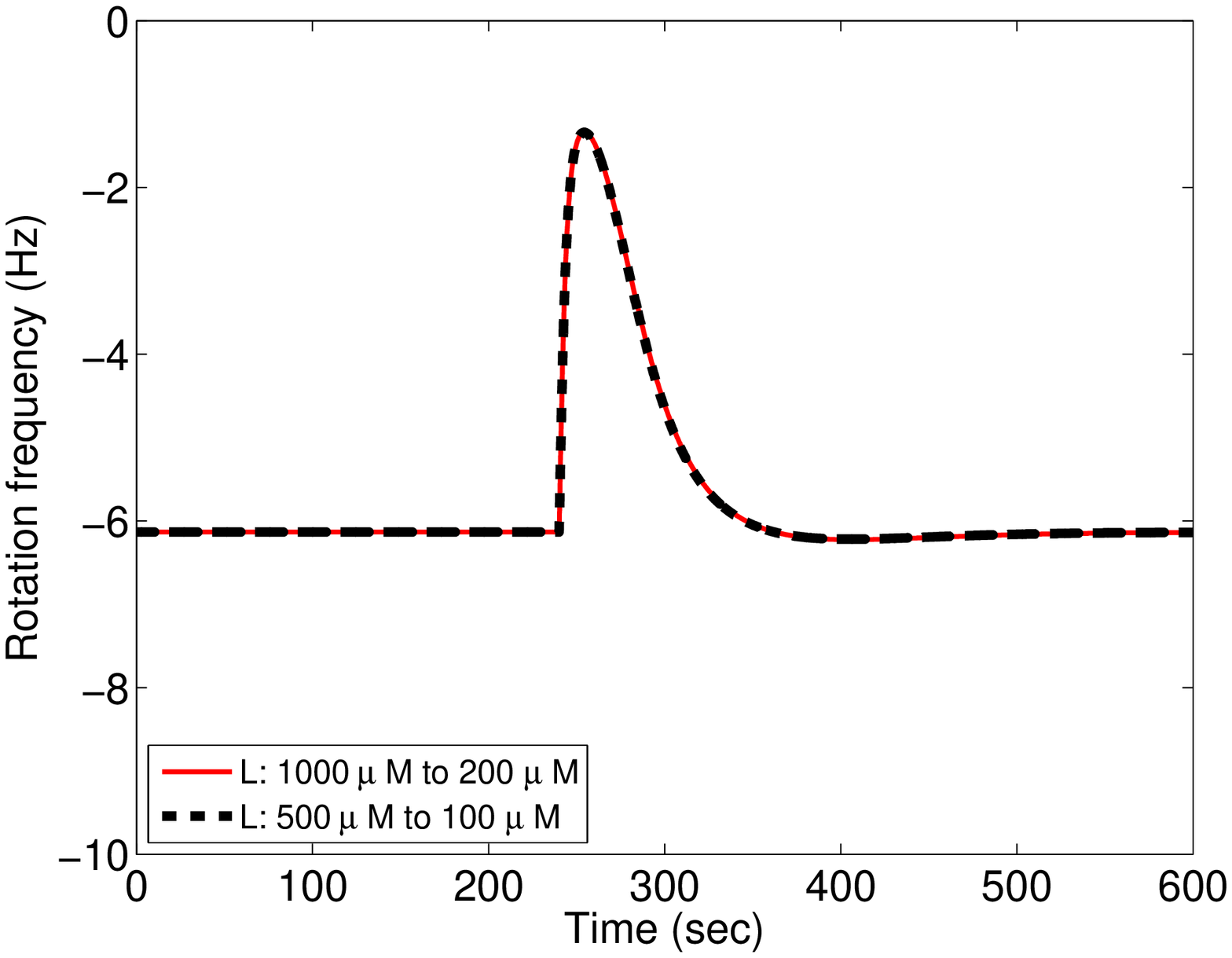}
\caption{Model I output in response to step changes in $L$ from
$L=1000 \mu$M to 200 $\mu$M and from $L=500 \mu$ M to $100
\mu$M}\label{Fig:model_I_fcd}\end{center}\end{figure}

\subsubsection{Model II}\label{Subsec:model II}


\noindent \textbf{Model structure:} Here, the model's internally
sensed ligands $\tilde{L}$ are related to $L$ via the differential
equation $\dot{\tilde{L}}=-\frac{1}{2}\tilde{L}+\frac{1}{2}L$.
Whilst the ligand concentrations $L,\tilde{L}$ modify receptor
activities, the cytoplasmic receptors are otherwise unregulated by
internal cell signals, and therefore $\tilde{L}$ is not a function
of $u$. A schematic is illustrated in Figure
\ref{Fig:model_II_schematic}, and a simulation of the model
together with the tethered cell trace to which the model was
fitted is shown in Figure \ref{Fig:model_II_assay}.  For
comparison, Figure \ref{Fig:model_II_assay} additionally shows a
simulation (with the same ligand input) of the model suggested in
\cite{plos}, which was fitted to the same tethered cell assay.
The root mean squared error between the output of Model II and the
tethered cell assay is 0.95, which, as with Model I, also compares
favorably to the corresponding error for the model in \cite{plos},
which is 1.27.


Note that if $L$ were to undergo a step change from $L=L_{a} \mu$M
to $L=L_{b} \mu$M and if the system is initially at steady state
(where $\tilde{L}(0)=L_{a}$), then $\tilde{L}$ would remain
confined to the set $[L_{a},L_{b})$. Therefore, for such a step
change, $K_{I} \ll L \ll K_{A}$ implies that $\tilde{K}_{I} \ll
\tilde{L} \ll \tilde{K}_{A}$.  Figure \ref{Fig:model_II_fcd} shows
that simulations of this model do show FCD in this input range, with similar
output traces obtained for the step changes in $L$ from $L=1000
\mu$M to 200 $\mu$M and from $L=500 \mu$ M to $100 \mu$M.

%

\noindent \textbf{Model parameters:}  $k_{R}=\tilde{k}_{R}=0.0057
\quad k_{B_{1}}=\tilde{k}_{B_{2}}=2.376 \quad k_{B_{2}}=2.970$.

\begin{figure}[h!]\begin{center}
\scalebox{0.5}{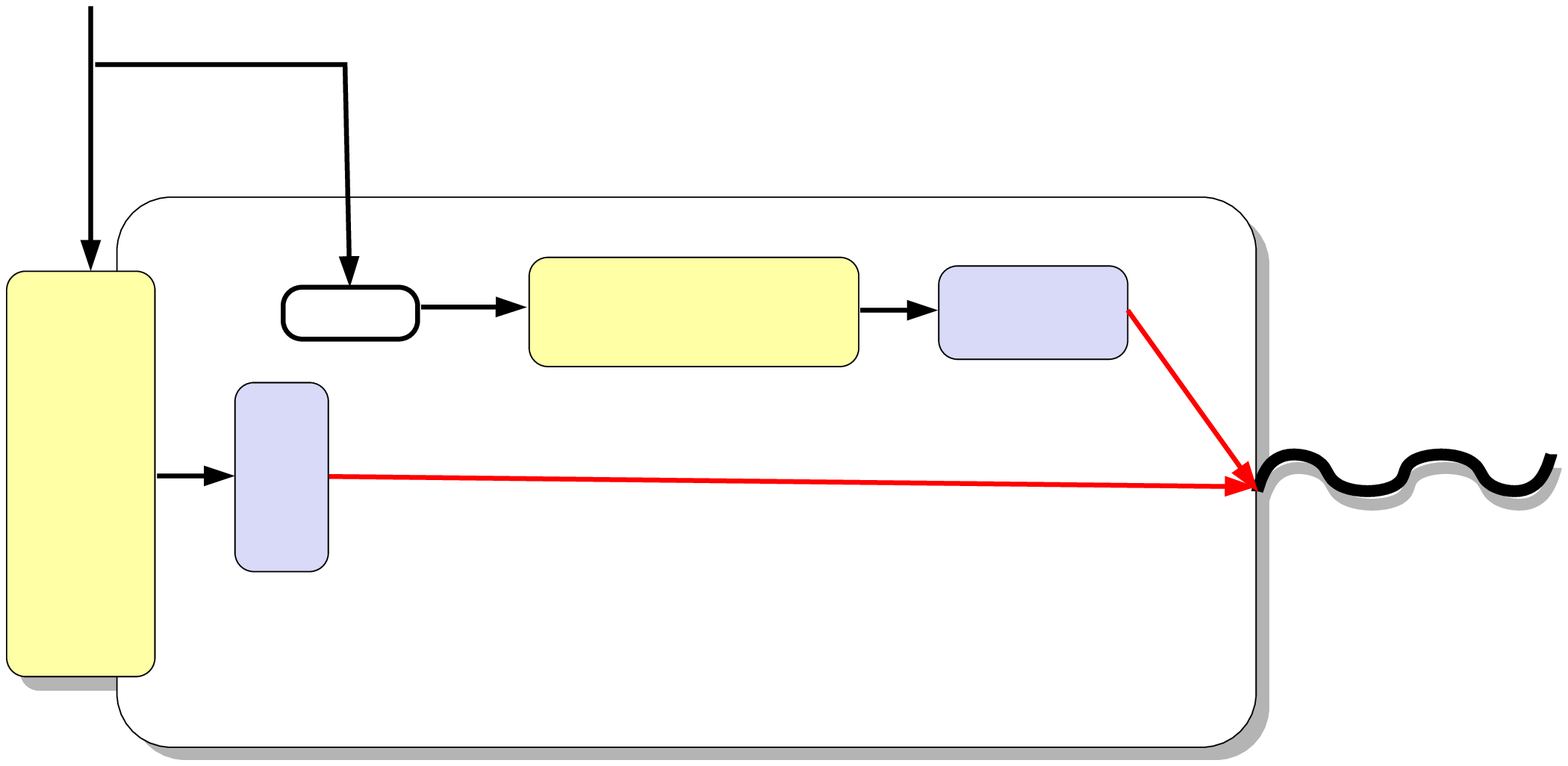}
\caption{Schematic of Model II.}
\label{Fig:model_II_schematic}\end{center}\end{figure}
\begin{figure}[h!]\begin{center}\includegraphics[width=0.7\textwidth]{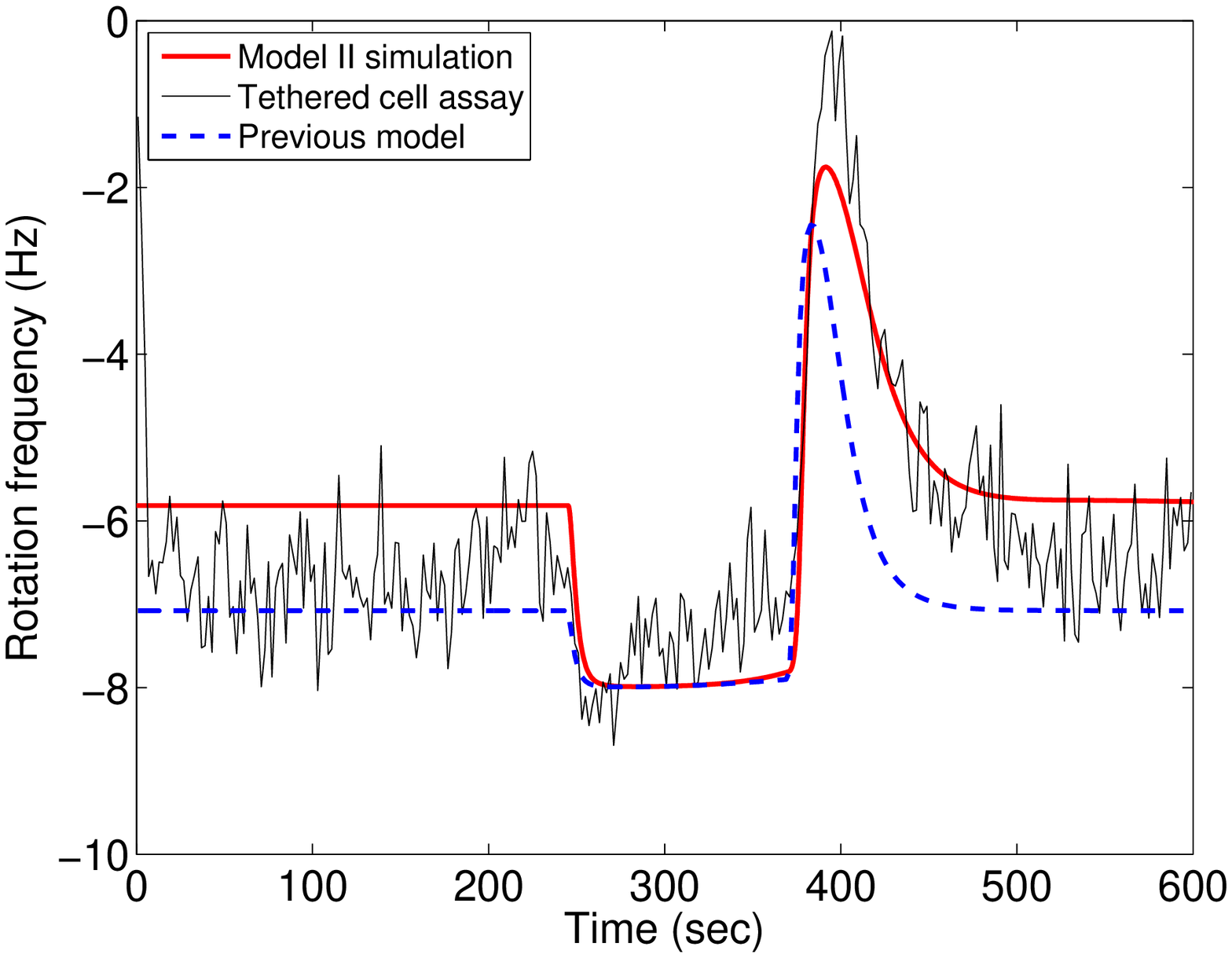}
\caption{Simulation of Model II (red) in response to a step rise
(at 245 seconds) and fall (at 370 seconds) in the ligand level $L$
from $L=0$ to $L=100$ and back to $L=0$, with a tethered cell
assay (black).  The dashed blue trace is a simulation of the
previously published model in \cite{plos} subject to the same
ligand input.}\label{Fig:model_II_assay}\end{center}\end{figure}

\begin{figure}[h!]\begin{center}\includegraphics[width=0.7\textwidth]{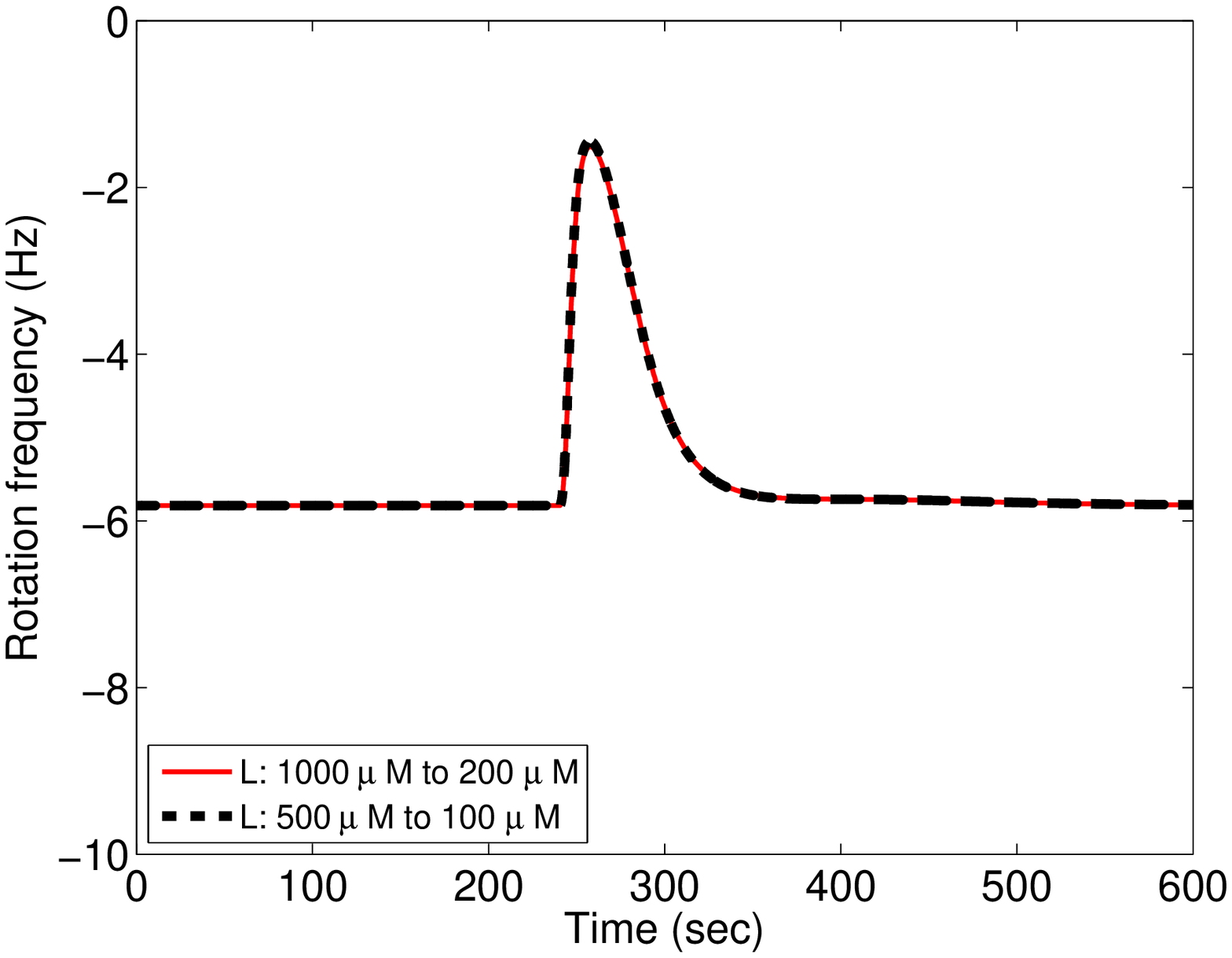}
\caption{Model II output in response to step changes in $L$ from
$L=1000 \mu$M to 200 $\mu$M and from $L=500 \mu$ M to $100
\mu$M}\label{Fig:model_II_fcd}\end{center}\end{figure}

\subsection{Future experiments for model invalidation}

The models presented above are two systems based on the assumptions of Section \ref{Sec:Assumptions} that reproduce the experimental data of \cite{plos}, but which additionally show FCD. By comparison, the model suggested in \cite{plos}, based on different receptor dynamics, does not exhibit FCD in response to the inputs used in the simulation in Figures \ref{Fig:model_I_fcd}, \ref{Fig:model_II_fcd} as shown in Figure \ref{Fig:model_plos_fcd}.

One important feature to note is that the FCD property is preserved
regardless of the exact dynamics in (\ref{Eq: signal transduction}) and regardless of the interactions between the receptors and the chemotaxis proteins, as long as the conditions of Theorem \ref{Thm: FCD} are satisfied.   Therefore, if the receptor dynamics model given in Section \ref{Sec:Assumptions} is accurate, then in the ligand concentration ranges $K_{I} \ll L \ll
K_{A}$ and $\tilde{K}_{I} \ll \tilde{L} \ll \tilde{K}_{A}$, FCD is a robust dynamic property of the chemotaxis system that should be observed in both wild type and in mutant strains of \emph{R. sphaeroides} that have chemotaxis protein deletions and over-expressions.

The above points suggest experiments in which the FCD dynamic phenotype
can be used to discriminate between Models I and II on the one hand, and the model suggested in \cite{plos} on the other:

\begin{itemize}
  \item If Models I and II are to invalidate that of \cite{plos} then the wild type bacterium, initially at steady state, should show near identical flagellar output responses to the step ligand inputs $L=1000 \mu$M to 200 $\mu$M and $L=500 \mu$ M to $100$.

\item Overexpressing the chemotaxis protein CheY$_{4}$ five fold was shown in \cite{BMC} to not destroy the chemotactic response of the bacterium.  Such a mutant strain should therefore, according to Models I and II, also exhibit FCD in response to a range of step changes in the external ligand concentration $L$. We can calculate this range for each of the two models as in Sections \ref{Subsec:model I} and \ref{Subsec:model II}. In Model I, the five-fold increase in CheY$_4$ means that FCD should be observed within the range $7.92 \tilde{K}_{I} \ll  L \ll
\tilde{K}_{A}$, whereas for Model II this range is $\tilde{K}_{I} \ll  L \ll
\tilde{K}_{A}$.

    \item If we define the adaptation time when the model is subject to a step decrease in ligand to be the time that it takes from the application of the step for the deviation of the flagellar rotation frequency from its steady-state value to fall to 25\% of its maximum, then under this definition, the adaptation times for Models I and II are 65 seconds and 62 seconds respectively, whereas for the model in \cite{plos}, the adaptation times are 266 seconds for the step ligand concentration decrease of 1000 $\mu$ M to 200 $\mu$ M, and 162 seconds for the step decrease of 500 $\mu$ M to 100 $\mu$M.  If Models I and II are to invalidate those in \cite{plos}, the experiments should yield approximately equal adaptation times in response to these two step changes in ligand concentration, and these adaptation times should be around 60 seconds.
\end{itemize}
Further model discrimination between Models I and II can be performed using the tools presented in \cite{BMC,plos}.

\begin{figure}[h!]\begin{center}\includegraphics[width=0.7\textwidth]{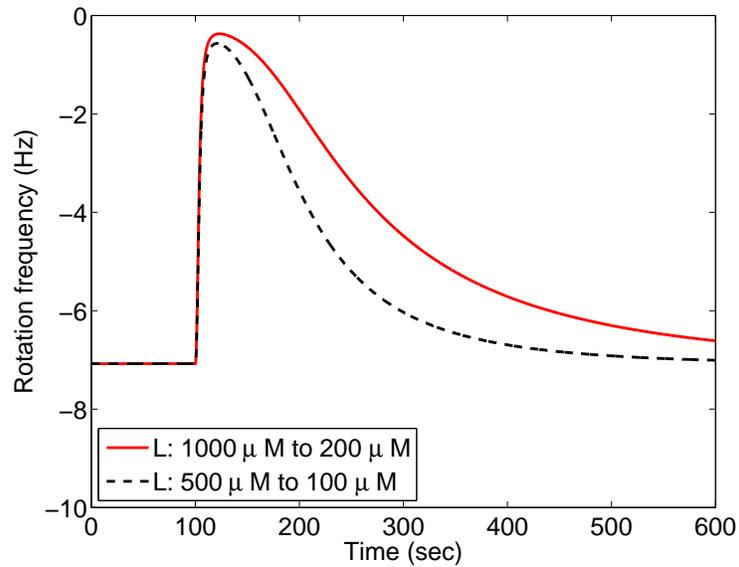}
\caption{Simulations of the model in \cite{plos}, subject to step
changes in $L$ from $L=1000 \mu$M to 200 $\mu$M and from $L=500
\mu$ M to $100
\mu$M}\label{Fig:model_plos_fcd}\end{center}\end{figure}

\section{Discussion}


The models presented herein differ from earlier \emph{R.
sphaeroides} chemotaxis models in two main respects: first, the
receptor dynamics are based on the MWC allosteric model.  This
model has been shown to be a fairly accurate representation of the
receptor dynamics in \emph{E. coli}.  The homologies between the
bacteria and the similarities between their overall chemotaxis
mechanisms give us reason to believe that the MWC model may, under
experimental testing, eventually prove to be a realistic way of
representing the \emph{R. sphaeroides} receptor dynamics.

The second point of departure of these models from earlier ones is
that the assumptions on the possible relationships between the
external and internal ligand concentrations are relaxed to admit
dynamic relations.  The motivation behind this model is to capture
any phase delays between sensed changes in the external ligand
concentration and the effect of such changes on the internal cell
environment.

The external-internal ligand relation of Model I closely follows
that of \cite{plos}. In effect, the activity of the cytoplasmic
cluster depends on the external ligand concentration, $L$, and,
indirectly, on the activity $a$ of the membrane cluster via the
phosphorylated chemotaxis proteins CheY$_3$-P and CheY$_4$-P, as
schematically illustrated in Figure \ref{Fig:model_I_schematic}.
On the other hand, the cytoplasmic receptor activity in Model II
does not depend on any chemotaxis proteins, and its sensed ligand
signals are merely phase-delayed versions of the external ligand
concentration.

As experimentally shown \cite{ArmitageNatureReview}, chemotaxis
requires CheY$_{6}$ \emph{and} one of either CheY$_{3}$ or
CheY$_{4}$, as deletion of either CheY$_{6}$ or both of
CheY$_{3_{p}}$ and CheY$_{4_{p}}$ destroys the chemotactic ability
of the bacterium.  In the models we present, this was captured by
the interaction of the three CheY proteins at the flagellum in
what is effectively an AND logic gate that will only be activated
if both CheY$_{6}$ and at least one of CheY$_{3_{p}}$ or
CheY$_{4_{p}}$ are present.  The signal transduction dynamics
(\cite{Porter,BMC}) show that CheY$_{3_{p}}$ and CheY$_{4_{p}}$
are solely phosphorylated by the membrane cluster, whereas
CheY$_{6}$ receives most of its phosphates from the cytoplasmic
cluster.  In essence, this structure means that there are
essentially two paths from the external ligands to the flagellum
that terminate at the AND gate: one path via the membrane cluster in
which CheY$_{3_{p}}$ and CheY$_{4_{p}}$ proteins convey the
signal, and one path via the cytoplasmic cluster, in which
CheY$_{6_{p}}$ conveys the signal.  This resembles a recurring
biochemical motif \cite{AlonBook}, and the selective advantage it bestows could be
improved \emph{energy taxis} \cite{ArmitageNatureReview} with
respect to simpler chemotaxis circuits such as that of \emph{E.
coli}.  The main feature of this improved pathway is that the
flagellar motion will only vary if both signalling paths from $L$
to the flagellum are activated.  Since the cytoplasmic cluster may
integrate un-modeled metabolic information from within the cell,
it would be important that any variation in flagellar activity
only results from a change in the metabolic state of the cell that
arises from a change in the local chemoeffector environment.  If
this is indeed the case then the signalling path from the
cytoplasmic cluster is only activated if the metabolic state of
the cell changes, whilst the signalling path from the membrane
cluster is only activated if the immediate chemical environment
changes.  Only if both are activated together would the cell
`know' that the change in its metabolic state is due to a change
in chemoeffector concentration, and only then would it change its
flagellar activity.


\subsection{Selective advantage of FCD}
Whether FCD bestows upon the bacterium a selective advantage or
simply arises as a by-product of the chemotaxis system's structure
is a question of interest.  The fact that FCD is present in
simpler chemotaxis circuits than that of \emph{R. sphaeroides}
(e.g. in \emph{E. coli}) suggests that the advantages gained by
having such a property would be independent of the complexity of
the bacterium's chemotaxis pathway. It may be that the metabolic
payoff to the bacterium of moving to more chemically favorable
regions depends on the relative chemical improvement in its
environment rather than the absolute change. A reason for this
could be that biasing its movement towards longer swims could be
metabolically costly for the bacterium, and moving in this way is
only worthwhile if the metabolic gain is significant.  A potential
disadvantage of FCD to the bacterium could be a high sensitivity
to small fluctuations in sensed ligand when the background ligand
concentration is low, due to the fact that the gain in the
flagellar rotation frequency would then be high. However, this
disadvantage is offset by the fact that FCD behavior only occurs
at background ligand concentrations significantly above a
threshold, given by $K_{I}$ in the models above.

\subsection{FCD as dynamic phenotype for model invalidation}

The models we have presented provide an example of how
dynamic phenotypes can be used to discriminate between competing
biochemical models.  Given two models of the same system, a
mathematical analysis can be used to identify regions in the
parameter and input spaces in which a certain qualitative dynamic
behavior, such as FCD, could be expected.  Ideally, this behavior
would be expected to be robust to any genetic mutations or
environmental conditions, and the conditions under which this
behavior would occur would be implementable experimentally. Model
discrimination can then be performed on the basis of whether or
not the system robustly reproduces the dynamic phenotype
experimentally. This differs from traditional forms of model
discrimination in that it can be used to discriminate between
different biological mechanisms, and can be used to identify
whether an observed phenomenon is due to the fine tuning of
biological parameters or due to a more fundamental structural
property of the system.

\end{document}